\documentclass[review,11pt,nonatbib]{elsarticle}

\usepackage[a4paper,margin=2.5cm]{geometry}
\usepackage{amsmath,amssymb,amsthm}
\usepackage{mathtools}
\usepackage{enumitem}
\usepackage{booktabs}
\usepackage{tikz}
\usetikzlibrary{arrows.meta,positioning}
\usepackage[natbibapa]{apacite}   
\renewcommand{\doiprefix}{https://doi.org/}

\usepackage[hidelinks]{hyperref}

\newcommand{\pipes}{\mathcal{P}}
\newcommand{\unitval}{v}
\newtheorem{theorem}{Theorem}
\newtheorem{proposition}[theorem]{Proposition}
\newtheorem{lemma}[theorem]{Lemma}
\newtheorem{corollary}[theorem]{Corollary}
\theoremstyle{definition}
\newtheorem{definition}[theorem]{Definition}
\newtheorem{example}[theorem]{Example}

\theoremstyle{remark}
\newtheorem{remark}[theorem]{Remark}

\newcommand{\R}{\mathbb{R}}
\newcommand{\Rp}{\mathbb{R}_{\ge 0}}
\newcommand{\Rr}{\mathcal{R}}
\newcommand{\wval}{w}

\makeatletter
\def\ps@pprintTitle{%
 \let\@oddhead\@empty
 \let\@evenhead\@empty
 \def\@oddfoot{\hfill\thepage}%
 \def\@evenfoot{\thepage\hfill}%
}
\makeatother

\begin{document}

\begin{frontmatter}

\title{Auction-based allocation of routes in capacity-constrained networks}

\author{Gergely Cs\'{a}ji}

\address{HUN-REN Institute for Computer Science and Control,\\  Kende u. 13-17., H-1111 Budapest;

Budapest University of Technology and Economics,\\  Műegyetem rkp. 3., H-1111, Budapest\\ \emph{csajigergely@sztaki.hu}}

\author{D\'{a}vid Csercsik \corref{cor1}}

\address{ELTE Centre for Economic and Regional Studies, Institute of Economics,\\ T\'oth K\'alm\'an u. 4., H-1097 Budapest\\  Tel.: +36-1 309 26 52  Fax: +36-1 319 31 36 \\
P\'{a}zm\'{a}ny P\'{e}ter Catholic University, Faculty of Information Technology and Bionics, Pr\'{a}ter~u. 50/A 1083 Budapest, Hungary \\
              Tel.: +36-1 886 47 00
              Fax: +36-1 886 47 24\\
              \emph{csercsik.david@krtk.elte.hu}}

\cortext[cor1]{Corresponding author}

\begin{abstract}
We study the route-set allocation problem between end nodes in capacity-constrained flow networks with transfer costs. We introduce an
\emph{origin--destination} (OD) bidding language, in which a participant bids
for bounded conic combinations of inlet--outlet--quantity licences and the
mechanism performs the routing through a network-flow optimising winner-determination
problem. We then study payment rules for the OD mechanism. Under a
\emph{restricted} model with openly and uniformly priced sources and
bidders who are indifferent to the source of the gas, the
Vickrey--Clarke--Groves (VCG) payment is core-stable; the key structural fact
is that in this regime some optimal aggregate flows never use a physical pipe
in both directions, so the winner-determination problem reduces to a single-commodity
network flow whose value is submodular in the set of active bidders. If
source access is allowed to depend on the individual consumer -- a realistic
feature, e.g.\ when gas deals are embedded in bilateral political
bargaining -- the reduction fails and VCG payments may lie strictly outside
the core. For this general model we propose a \emph{dual (nodal) pricing} rule
derived from the shadow prices of the flow-conservation constraints, and prove
it is always core-stable, at the cost of dominant-strategy incentive
compatibility. Illustrative examples accompany each result.
\end{abstract}


\end{frontmatter}

\section{Introduction}\label{sec:intro}

\subsection{Motivation and background}

The allocation of scarce transport capacity in a network, where the medium is steerable -- gas pipelines, freight corridors, data links -- is a mechanism
design problem in which the goods (edge capacities) are complements: a participant who is required to ship goods to their desired destination
values a capacity product only as part of a complete route from a source to
its consumption point. The same structure -- endpoint demand over a shared,
capacitated network -- underlies auction designs for bandwidth and inter-domain
traffic \citep{Belzarena2009,MailleTuffin2007} and for railway
paths \citep{Stojadinovic2019}, as discussed in Section~\ref{sec:related}.
The current European practice for natural gas pipelines allocates capacity
through simultaneous ascending clock auctions on the PRISMA platform, a design
whose shortcomings under a dominant supplier are well
documented \citep{TakacsneToth2017}.

As an alternative, the concept of \emph{convex combinatorial
auction} (CCA) has been proposed \citep{Csercsik2022}, which exploits the divisibility of pipeline capacity: bidders
submit route--quantity bids, and the clearing assigns each participant a convex
combination of its submitted bids by solving a single linear program, with
payments determined by the Vickrey--Clarke--Groves (VCG) principle. Relative to
the ascending clock auction, the CCA allocates only usable capacity, attains higher
resulting utility, allows a more efficient utilisation of the available infrastructure and reduces the bidding burden.

The CCA framework is characterised by one important structural limitation that motivates the
present paper. Each participant must determine the set of
routes on which it is willing to bid in advance; the mechanism then optimises only over
convex combinations of these predefined routes. In the example networks
of \citet{Csercsik2022} the route set is fixed a priori, and the paper itself
lists ``the routes considered in the auction'' as an exogenous input and flags
their determination as future work. Because the number of source--sink routes
in a network grows exponentially with its size, complete enumeration is
infeasible for realistic instances, and any incomplete enumeration can, as we
show, induce inefficient results in terms of welfare.

\subsection{Contributions}

The paper is organised around two models of source access, treated in two main
parts.

\begin{enumerate}[leftmargin=1.6em]
\item \textbf{The general model (Section~\ref{sec:general}).} When source access or price may depend on the consumer, bidders evaluate each source differently. Accordingly, we introduce an origin--destination (OD) bidding
language of (inlet, outlet, quantity) licences, with the WDP -- a
network-flow linear program -- performing the routing; this removes the route
enumeration of the convex combinatorial auction of \citet{Csercsik2022}, whose
fixed route sets can be strictly suboptimal (Section~\ref{subsec:inefficiency},
Example~\ref{ex:diamond}). In this model, VCG payments can lie strictly outside
the core (Example~\ref{ex:series}). More generally, such a failure requires
a departure from source indifference; dedicated access and inlet-specific
quantity caps are two ways this can occur (Proposition~\ref{prop:necessity}). We therefore propose a
\emph{dual (nodal) pricing} rule read off the flow-conservation shadow prices and
prove it is always core-stable (Theorem~\ref{thm:dual-core}), at the cost of
dominant-strategy incentive compatibility -- an unavoidable trade-off -- but
core-selecting in the spirit of \citet{DayMilgrom2008}.

\item \textbf{The restricted model (Section~\ref{sec:restricted}).} Sources are
openly and uniformly priced and the good is homogeneous, so a consumer is
indifferent to \emph{which} source supplies it. A bidder then submits atomic
\emph{quantity--price} licences describing its inverse demand at its own sink,
and the winner-determination problem (WDP) is a single-commodity min-cost flow.
We prove that the coalition value is submodular in the set of bidders, so the
Vickrey--Clarke--Groves (VCG) payment is core-stable
(Proposition~\ref{prop:vcg-core}). The proof is elementary: an exchange argument
(Lemma~\ref{lem:noopposing}) shows optimal flows can be chosen such that they do not oppose on any edges, and a linear-programming (LP)
duality argument on the node potentials (Lemma~\ref{lem:submod}) yields
submodularity directly. We also show that this positive result genuinely needs
source indifference, and inlet-specific quantity caps may undermine
submodularity (Examples~\ref{ex:oppose}--\ref{ex:submodfail}).

\end{enumerate}

The remainder of the paper is organised as follows.
Section~\ref{sec:model} introduces the network model and
Section~\ref{sec:prelim} the auction-theoretic preliminaries (core and VCG).
Section~\ref{sec:general} formulates the OD bidding language and the general
winner-determination problem. Section~\ref{sec:restricted} studies the
restricted model and establishes the core stability of VCG.
Section~\ref{sec:general-results} compares the OD mechanism with route
enumeration, exhibits failures of VCG beyond source indifference, and
introduces dual pricing. Section~\ref{sec:discussion} discusses incentives,
counter-directed flows and the scope of the results, and
Section~\ref{sec:conclusion} concludes.

\subsection{Related work}\label{sec:related}

Auction-based allocation of divisible capacity on a network, with endpoint (rather than link-by-link) bidding and prices derived from the optimisation
duals, has precedent outside the natural gas literature. The bandwidth-auction design of \citet{Belzarena2009} can be regarded as closest to the present
mechanism: users place a single bid for an end-to-end connection ``oblivious to the internal topology'', the operator solves a capacity-constrained allocation, and per-route prices are read off the Lagrange multipliers of the link constraints.
This approach structurally aligns the endpoint bidding language and dual pricing proposed in the current article; the underlying idea of pricing shared network capacity by the shadow prices (Lagrange multipliers) of the link constraints is the network-utility-maximisation tradition of \citet{kelly1998rate}. Two features separate their
setting from ours and locate our contribution. First, the good considered in \citet{Belzarena2009} is
indivisible (a circuit), so their allocation is an integer program whose convex
relaxation is \emph{not} exact. Rounding must be applied, and no exact dual price exists.
In our model capacity is divisible, the WDP is an exact linear program, and its
duals are genuine competitive-equilibrium prices -- which is precisely what
makes the core-stability guarantee available. Second, they adopt a first-price
rule and explicitly decline VCG, observing that on a scarce link ``the
generalized Vickrey auction would charge a price equal to the highest bid left
out $\ldots$ hence the network receives zero revenue''. This phenomenon is related to the generally known low-revenue
pathology of VCG, which may imply that the VCG-induced outputs may lie outside the core \citep{DayMilgrom2008}.
On the other hand, the application of a first-price rule (i.e.\ a ``pay-as-bid'' approach) has straightforward drawbacks related to strategic bidding \citep{tierney2008uniform}.
\citet{Belzarena2009} furthermore argue that computing the VCG prices is demanding, since an integer program must be solved in the case of each player.
Scepticism about VCG for network pricing is independently documented by \citet{MailleTuffin2007}, who argue it ``can hardly be applied'' to inter-domain and ad hoc networks for two reasons: with capacity
constraints and single-path routing the allocation (and the $N{+}1$ pricing
optimizations) is NP-hard, and without them VCG violates budget balance,
requiring the operator to inject money. The divisible, multi-path
(network-flow) formulation presented in this paper dissolves this first objection -- the WDP is a polynomial linear program, so the VCG allocation and payments are all tractable.

In addition, in contrast to \citet{Belzarena2009}, \citet{MailleTuffin2007} and \citep{Csercsik2022}, we also consider the question of stability and propose a core-stable payment rule.
The impossibility framing of \citet{MoulinShenker2001} (efficiency, incentive compatibility,
individual rationality and budget balance cannot hold jointly)
is consistent with our design: the dual-price rule relaxes dominant-strategy
incentive compatibility, not budget balance or individual rationality.

Auction and combinatorial-market designs for network-path allocation also
appear in other regulated infrastructures, notably railway capacity
allocation \citep{Stojadinovic2019,Shahmohammadi2024}, where competing operators
bid for train paths through a shared, capacitated network.
In the setting of railway capacity allocation, a train path is a rigid sequence of track segments and time slots,
far less substitutable than fungible gas flow, so the route bundle is characterised by an intrinsic complementarity and indivisibility.
In contrast, our contribution is designed for fungible and divisible commodities like gas or bandwidth.
In the following we will primarily consider the context of natural gas networks, and define the assumptions of the proposed model accordingly.

\section{Network model}\label{sec:model}

We work with the network abstraction of the model used in \citet{Csercsik2022}. Let $G=(N,E)$ be a
directed graph: nodes $N$ represent market areas -- the entry--exit zones within
which capacity is booked and traded, whose efficiency is analysed
by \citet{keller2019mergers} -- and each edge $e\in E$
represents a directed capacity product with maximal transfer capacity
$\bar q_e\ge 0$ and per-unit transfer cost $c^t_e\ge 0$. A subset
$N_s\subseteq N$ of nodes are \emph{sources}; source $s\in N_s$ can inject gas
at per-unit cost $c^s_s\ge 0$. The good is \emph{homogeneous} -- gas of identical
quality (calorific value) regardless of source -- so a unit delivered to a
consumer is equally valuable whatever inlet it entered through; sources differ
only in price and location, not in the good itself. The remaining nodes host
\emph{consumers} (bidders) or are purely intermediate. Each bidder $i$ consumes
at some node $n_i$ (no consumption node is shared by two bidders) and, in
line with the piecewise-constant inverse-demand curves used
in \citet{Csercsik2022}, has a concave, piecewise-linear, nondecreasing
willingness-to-pay.

Let $\pipes$ denote the set of physical pipes. Each pipeline is modelled by a pair of oppositely directed edges with
independent nonnegative capacities; a physical flow uses at most one direction
of a pipe at a time. We distinguish two capacity conventions and treat both:

\begin{description}[leftmargin=1.4em]
\item[Directed-arc (no netting).] Each directed edge has its own capacity
constraint on its own nonnegative flow.
\item[Netting.] Counter-directed flows on the two arcs of a pipe cancel, so the
capacity constraint binds the \emph{net} flow, as in Eq.~(7)
of \citet{Csercsik2022}.
\end{description}
Under the directed-arc convention, opposite directional capacity bookings are treated as independent and may coexist; under netting, they are aggregated into physical net flow.
\subsection{Two source regimes}
\label{subsec_two_source_regimes}

The distinction that organises the paper is how sources may be accessed by
consumers.

\begin{definition}[General model with (potentially) dedicated sources]\label{def:general}
The accessibility and/or price of a source may depend on the individual
consumer: bidder $i$ accesses source $s$ at price $c^{s}_{s,i}\in\Rp\cup\{+\infty\}$
($+\infty$ encoding inaccessibility). The restricted model is the special case
$c^s_{s,i}=c^s_s$ for all $i$ with, in addition, source indifference.
\end{definition}

Consumer-dependent access is not a purely artificial idea; it is motivated by
practical relevance. In natural gas trading, long-term supply relationships,
take-or-pay contracts, and the embedding of gas deals in bilateral political
bargaining all affect the effective price at which a given importer can draw from a
given source specific to that importer in the period for which capacity is
allocated; the geopolitical discussion of \citet{Bilgin2009} is precisely about
such consumer-specific access. Our results identify source indifference,
rather than divisibility or network size, as the boundary of the positive
guarantee: consumer-specific access and inlet-specific caps are distinct ways
to cross that boundary.

\begin{definition}[Restricted model $\Rr$]\label{def:restricted}
Every node is source-only, consumer-only, or intermediate; every source is
openly accessible to every bidder at a \emph{consumer-independent} posted price
$c^s_s$; and, the good being homogeneous, each bidder is \emph{indifferent to the
source} of its gas -- it values only the quantity delivered to its own sink, not
the inlet through which it entered. This is the regime of \citet{Csercsik2022}.
\end{definition}

\section{Preliminaries: core and VCG}\label{sec:prelim}

We recall the two solution concepts of auction theory used in the paper. Let the set of players be
$I\cup\{0\}$, where $I$ is the set of bidders and $0$ the auctioneer (holder of
all capacity). An \emph{outcome} is an ordered pair of an allocation ($\mathcal{A}$) and a payment vector $P$. The payoff or utility of player
$i$ is $u_i=V_i(\mathcal{A})-p_i$, where $V_i(\mathcal{A})$ is the valuation of items received by player $i$ in allocation $\mathcal{A}$.

In our case, a bidder's total payment $p_i$ may be decomposed into three streams with distinct
recipients, and only one of them contributes to the income of the auctioneer. The
\emph{source payments} represent income to the gas suppliers (or, for positioned
gas, are the shipper's own sunk acquisition cost); the \emph{transfer payments}
are paid to the local transmission system operators (TSOs), who under the EU
third-party-access principle set and collect these regulated tariffs on their own
infrastructure; and the \emph{capacity payments} $CC_i$ -- the charge for the
scarce capacity licences -- accrue to the auctioneer. Source and transfer
payments are therefore pass-throughs, not auctioneer income; their explicit form
in terms of the source costs $c^s$ and transfer costs $c^t$ is given with the
winner-determination problem in Section~\ref{subsec:language} (eq.~\eqref{eq:wdp}).
Accordingly the auctioneer's payoff is
\[
   u_0=\sum_{i}CC_i,
\]
the sum of capacity payments only, in line with the income of the auctioneer in
the model of \citet{Csercsik2022}. As shown in Section~\ref{subsec:dual}, under the
dual price this equals the congestion rent
$\sum_{a\in\pipes}(\lambda^{+\star}_a\bar q^+_a+\lambda^{-\star}_a\bar q^-_a)\ge 0$,
so the auctioneer earns exactly the scarcity value of the binding capacity
constraints and nothing in the absence of congestion. Suppliers and TSOs are
paid precisely their costs and are treated as cost sinks rather than players in the proposed model.
Consequently, the coalition value $\wval(S)$ below is equal to the \emph{net welfare} (valuations minus source and transfer costs), which is exactly the objective of the WDP.

\begin{definition}[Coalition value]
For $S\subseteq I$, the coalition value $\wval(S)$ is the optimum of the
winner-determination program in force -- see eq. \eqref{eq:wdp-r} in the restricted model
(Section~\ref{sec:restricted}), and eq. \eqref{eq:wdp} in the general model
(Section~\ref{sec:general}) -- the maximal net welfare attainable when the
bidders in $S$ are served over the full network, with $\wval(S)=0$ if the
auctioneer is absent. We set $\wval(\emptyset)=0$.
\end{definition}

\begin{definition}[Core]\label{def:core}
An outcome with payoffs $(u_i)_{i\in I},u_0$ is in the \emph{core} if it is
feasible ($\sum_{i\in I}u_i+u_0=\wval(I)$), individually rational ($u_i\ge 0$),
and \emph{unblocked}: for every coalition $S\subseteq I$,
\[
   u_0+\sum_{i\in S}u_i \;\ge\; \wval(S).
\]
Equivalently, no coalition of the auctioneer with a set of bidders could, on its
own, guarantee its members more than they receive.
\end{definition}

\begin{definition}[VCG payment]\label{def:vcg}
Under the VCG (Clarke pivot) rule \citep{Vickrey1961,Clarke1971,Groves1973}, the
allocation maximises total reported value and bidder $i$ pays the externality it
imposes,
\[
   p_i^{\mathrm{VCG}} \;=\; \wval(I\setminus\{i\}) \;-\;
      \bigl[\wval(I)-V_i(\mathcal{A}_i)\bigr],
\]
the difference between the others' best value without $i$ and their value with
$i$ present. VCG is efficient and dominant-strategy incentive compatible (DSIC).
\end{definition}

A classical sufficient condition links the two. If $\wval$ is
\emph{submodular} over bidders -- equivalently, bidders are substitutes in the
sense of \citet{AusubelMilgrom2002,AusubelMilgrom2006} -- then the VCG payoff
vector lies in the core and is its bidder-optimal point; see also
\citet{bikhchandani2002package} for the underlying linear-programming
characterisation. Conversely, when $\wval$ is not submodular,
VCG revenue can be so low that a coalition blocks, and VCG leaves the core --
the phenomenon that motivates the concept of core-selecting auctions \citep{DayMilgrom2008}. 

\section{The general model: OD bidding and winner determination}
\label{sec:general}

We first formulate the general allocation problem, allowing source access
and source prices to depend on the bidder.  The restricted model will then be
obtained by imposing source indifference.

\subsection{The OD bidding language}
\label{subsec:language}

Instead of enumerating routes between potentially multiple sources and a
consumer's destination, a bidder describes only where gas may enter and where
it is to be delivered; the mechanism chooses the route.

\begin{definition}[OD licence]
An \emph{atomic licence} of bidder $i$ is a tuple
$b=(o_b,q_b,\unitval_b)$: an inlet node $o_b$, the bidder's outlet
$t_b=n_i$, a maximum quantity $q_b\ge0$, and a gross per-unit value
$\unitval_b\ge0$. Accepting $b$ at level $d_b\in[0,q_b]$ injects $d_b$ units at
$o_b$, delivers them at $n_i$, and contributes $\unitval_bd_b$ to gross value.
\end{definition}

\begin{definition}[Bid region]\label{def:bidregion}
Bidder $i$ submits a finite set $\mathcal B_i$ of atomic licences and an
overall cap $Q_i^T$. Its accepted vector ranges over the bounded polytope

\[
 P_i=\Bigl\{d_i\in\Rp^{\mathcal B_i}:d_b\le q_b\ \forall b,\quad
       \sum_{b\in\mathcal B_i}d_b\le Q_i^T\Bigr\}.
\]
The reported gross valuation is
$V_i(d_i)=\sum_{b\in\mathcal B_i}\unitval_bd_b$.
\end{definition}

\begin{example}\label{ex:license}
For a bidder consuming at node $C$, the statement ``I want up to $2$ units
from $A$ to my node $C$, worth $\unitval_1$ each; or up to $3$ units from $B$
to $C$, worth $\unitval_2$ each; but at most $4$ units in total'' corresponds to
the licences $(A,2,\unitval_1)$ and $(B,3,\unitval_2)$ with $Q^T=4$. The bidder names no path.
\end{example}

Source access is encoded in $\mathcal B_i$: a bidder that cannot use
source $s$ names no license originating there.

\subsection{General winner determination}

Recall that $\pipes$ is the set of physical pipes. Give each pipe
$a\in\pipes$ an arbitrary reference orientation
$a^+=(u_a,v_a)$ and reverse orientation $a^-=(v_a,u_a)$, with directional
capacities $\bar q_a^+,\bar q_a^-$ and nonnegative costs
$c_a^{t,+},c_a^{t,-}$. Write $f_{i,a}^+,f_{i,a}^-\ge0$ for bidder $i$'s
directional flows and $A$ for the incidence matrix of the resulting directed
arc set $\vec E=\{a^+,a^-:a\in\pipes\}$, which is the directed edge set $E$ of Section~\ref{sec:model} with its orientation made explicit. In the netting convention,
$g_a^+,g_a^-$ form a nonnegative split of the aggregate net flow. For a
coalition $S\subseteq I$, the general WDP is

\begin{equation}\label{eq:wdp}
\begin{aligned}
\wval(S)=\max_{d,f,g}\quad
 &\sum_{i\in S}V_i(d_i)
   -\sum_{a\in\pipes}\bigl(c_a^{t,+}g_a^++c_a^{t,-}g_a^-\bigr)
   -\sum_{i\in S}\sum_{b\in\mathcal B_i}c^s_{o_b,i}d_b\\
\text{s.t.}\quad
 &\sum_{a\in\vec E}A_{n,a}f_{i,a}
    =\mathbf 1[n=n_i]\sum_{b\in\mathcal B_i}d_b
      -\!\!\sum_{b\in\mathcal B_i:o_b=n}\!d_b
       &&\forall i\in S,\ \forall n \quad(\pi_{i,n}),\\
 &\sum_{i\in S}\bigl(f_{i,a}^+-f_{i,a}^-\bigr)=g_a^+-g_a^-
       &&\forall a\in\pipes \quad (\tau_a),\\
 &g_a^+\le\bar q_a^+ &&\forall a\in\pipes\quad(\lambda_a^+),\\
 &g_a^-\le\bar q_a^- &&\forall a\in\pipes\quad(\lambda_a^-),\\
 & d_b\le q_b
&& \forall i\in S,\ \forall b\in\mathcal B_i
\quad(\alpha_b),\notag\\
& \sum_{b\in\mathcal B_i}d_b\le Q_i^T
&& \forall i\in S\quad(\beta_i),\notag\\
 &d,f,g^+,g^-\ge0.
\end{aligned}
\end{equation}

Thus netting is imposed on the two directions of the \emph{same physical
pipe}, not on a single directed arc. The directed-arc convention is obtained
from \eqref{eq:wdp} by deleting $g^+,g^-$ and the net-flow identity, imposing

\[
 \sum_{i\in S}f_{i,a}^+\le\bar q_a^+,
 \qquad
 \sum_{i\in S}f_{i,a}^-\le\bar q_a^-
 \qquad(a\in\pipes),
\]
and replacing the transfer-cost term by
$\sum_{a,\sigma\in\{+,-\}}c_a^{t,\sigma}\sum_i f_{i,a}^{\sigma}$.
The conservation equation says exactly that bidder $i$ injects $d_b$ at each
licence inlet and withdraws $D_i=\sum_{b\in\mathcal B_i}d_b$ at its sink $n_i$. Both capacity conventions give
linear programs. LP~\eqref{eq:wdp} defines $\wval(S)$ in the general
model; a coalition without the auctioneer controls no capacity and has value
zero.

\begin{remark}
Consider a route-CCA of \citet{Csercsik2022} defined on the same network, with the same source
access, costs, capacity convention, and bidder quantity restrictions as the OD
WDP. If every enumerated route is a valid path from its stated inlet to the
bidder's sink, then every allocation feasible for that route-CCA can be embedded
in the OD WDP by routing each accepted quantity along its enumerated path.
Consequently, the route-CCA feasible set is contained in the OD feasible set
and
\[
    \wval(I)\ge \wval^{\mathrm{CCA}}(I).
\]
Conversely, if the route-CCA enumerates all feasible inlet--sink paths, then an
optimal OD flow can be decomposed into path flows for the individual bidders.
Since transfer costs are nonnegative, flow circulations that do not contribute
to any delivery can be omitted without decreasing welfare. Each remaining path
component corresponds to an enumerated route of the route-CCA, and assigning its
flow value to that route reproduces the same accepted quantities, network loads,
and welfare. Hence every optimal OD allocation has an equivalent allocation in
the completely enumerated route-CCA, so
\[
    \wval(I)\le \wval^{\mathrm{CCA}}(I).
\]
Therefore, under complete route enumeration,
\[
    \wval(I)=\wval^{\mathrm{CCA}}(I).
\]
Thus the advantage of the OD formulation is not a relaxation of the allocation
problem: it represents the complete-route optimum without requiring the bidders
to enumerate an exponentially large family of paths.
\end{remark}

\section{The restricted model: quantity--price licences and core-stable VCG}
\label{sec:restricted}

Throughout this section we are in the restricted model $\Rr$
(Definition~\ref{def:restricted}): sources are open and uniformly priced, and
each bidder is indifferent to which source supplies its (homogeneous) gas.

\subsection{Quantity--price licences}\label{subsec:qp}

\begin{definition}[Atomic licences and inverse demand]\label{def:qp}
In $\Rr$ a bidder $i$ (consuming at $n_i$) submits a finite set $\mathcal B_i$ of
atomic \emph{quantity--price} licences, each a pair $b=(q_b,\unitval_b)$: up to $q_b$
units delivered at $n_i$, worth $\unitval_b$ each. With acceptance levels
$d_b\in[0,q_b]$ and total delivery $D_i=\sum_{b\in\mathcal B_i}d_b$, the reported
value is $V_i(d_i)=\sum_{b\in\mathcal B_i}\unitval_b d_b$. Ordering the licences by
decreasing $\unitval_b$ makes $V_i$ a concave, piecewise-linear, nondecreasing function
of $D_i$ -- the bidder's inverse demand curve. We write
$P_i=\{d_i\ge0:d_b\le q_b\ \forall b,\
\sum_{b\in\mathcal B_i}d_b\le Q_i^T\}$. If no additional total cap is
desired, set $Q_i^T=\sum_{b\in\mathcal B_i}q_b$, so the last inequality is
redundant.
\end{definition}


\subsection{Why source indifference matters}\label{subsec:necessity-indiff}

Open, uniformly priced sources are not by themselves enough for the results
below: if a bidder could attach \emph{inlet-specific} quantity caps to its
licences, both the no-opposing-flows lemma  (Section~\ref{subsec:noopp}) and submodularity would fail. Source
indifference (no per-inlet caps) is exactly what rules this out.

\begin{example}[Per-source caps may induce opposing flows]\label{ex:oppose}
Let a single bidirectional pipe join $A$ and $B$, with a zero-cost source at each
end. Bidder $1$ consumes at $A$ but may draw at most $1/4$ from the source at $A$
and up to $1$ from the source at $B$; bidder $2$ is the mirror image at $B$. Each
wants one unit. Serving both requires bidder $1$ to import at least $3/4$
of a unit $B\to A$ and bidder $2$ to import at least $3/4$ of a unit $A\to B$,
so the pipe carries flow in \emph{both}
directions at the optimum. The source-swap that would remove the opposition --
reroute bidder $1$ to draw from $A$ -- is infeasible, as it violates bidder
$1$'s cap of $1/4$ at $A$.
\end{example}

\begin{example}[Submodularity fails under per-source caps]\label{ex:submodfail}
Take sources $s_A,s_B$ and arcs $s_A\to A\xrightarrow{e_1}B\xrightarrow{e_2}T$
and $s_B\to B$, with $e_1,e_2$ of capacity $1$ (Figure~\ref{fig:ce}); all costs
zero. Bidder $Y$ consumes at $B$; bidders $X$ and $Z$ consume at distinct nodes $t_X$ and $t_Z$, reached from $T$ by uncapacitated zero-cost arcs, so their gas
must cross $e_2$. Each wants at most one unit at value $1$, with inlet-specific
caps $X$: at most $1$ from $s_A$ and $\tfrac14$ from $s_B$; $Y$: same as $X$;
$Z$: at most $\tfrac14$ from $s_A$ and $1$ from $s_B$. Solving the flow problem per
coalition gives $\wval(X)=1$, $\wval(X,Z)=1$, $\wval(X,Y)=\tfrac32$,
$\wval(X,Y,Z)=2$, so
\[
   \wval(X,Y)+\wval(X,Z)=\tfrac52 \;<\; 3=\wval(X,Y,Z)+\wval(X),
\]
violating submodularity: the caps make $Y$ and $Z$ complements.
\end{example}

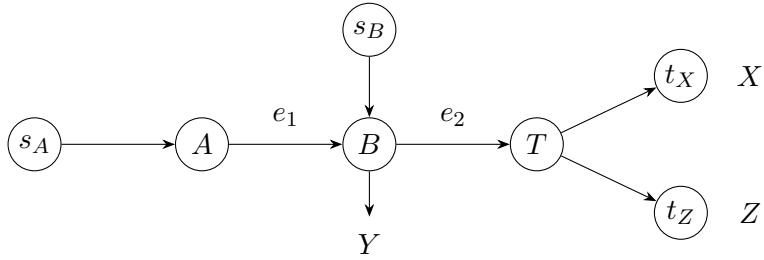
\begin{figure}[h]
\centering
\begin{tikzpicture}[>=Stealth,node distance=15mm,
  every node/.style={circle,draw,minimum size=7mm,inner sep=0pt}]
  \node (sA) {$s_A$};
  \node (A) [right=of sA] {$A$};
  \node (B) [right=of A] {$B$};
  \node (T) [right=of B] {$T$};
  \node (sB) [above=8mm of B] {$s_B$};
  \draw[->] (sA) -- (A);
  \draw[->] (A) -- node[above,draw=none,fill=none,inner sep=1pt]{$e_1$} (B);
  \draw[->] (B) -- node[above,draw=none,fill=none,inner sep=1pt]{$e_2$} (T);
  \draw[->] (sB) -- (B);
  \node[draw=none] (Y) [below=6mm of B] {$Y$};
  \draw[->] (B) -- (Y);
  \node (tX) [above right=4mm and 14mm of T] {$t_X$};
  \node (tZ) [below right=4mm and 14mm of T] {$t_Z$};
  \draw[->] (T) -- (tX);
  \draw[->] (T) -- (tZ);
  \node[draw=none] [right=2mm of tX] {$X$};
  \node[draw=none] [right=2mm of tZ] {$Z$};
\end{tikzpicture}
\caption{Example~\ref{ex:submodfail}. With inlet-specific caps the coalition
value is not submodular; source indifference removes the effect.}
\label{fig:ce}
\end{figure}

Both failures come from the same feature -- an inlet-specific cap forces a bidder
onto a particular source, manufacturing a counter-directed import
(Example~\ref{ex:oppose}) or a complementarity between otherwise-independent
bidders (Example~\ref{ex:submodfail}). Neither is present under source
indifference, which is why Definition~\ref{def:qp} admits only sink-side caps.

\subsection{Existence of optimal flows that do not oppose}\label{subsec:noopp}

\begin{lemma}[No opposing flows]\label{lem:noopposing}
In $\Rr$, every netting-feasible allocation can be converted into a
directed-arc feasible allocation with the same deliveries and weakly greater
welfare, in which no pipe carries flow in both directions.
\end{lemma}

\begin{proof}
Fix a netting-feasible allocation and a physical pipe $a$. Write its
gross directional flows as $x_a^+,x_a^-\ge0$ and its net flow as
$z_a=x_a^+-x_a^-$. Cancel
$\delta_a=\min\{x_a^+,x_a^-\}$ and retain

\[
 \widetilde x_a^+=[z_a]_+,
 \qquad
 \widetilde x_a^-=[-z_a]_+.
\]
At the two endpoints, the contribution to net inflow is unchanged; hence
every node balance is unchanged. The netting bounds
$z_a\le\bar q_a^+$ and $-z_a\le\bar q_a^-$ imply the separate directed bounds
$\widetilde x_a^+\le\bar q_a^+$ and
$\widetilde x_a^-\le\bar q_a^-$. Nonnegative transfer costs weakly decrease
(and under the net-flow cost convention are unchanged). Most importantly,
the source injections, accepted quantities $d_i$, and deliveries $D_i$ are
untouched, so every
atomic cap $d_b\le q_b$ and total cap $\sum_b d_b\le Q_i^T$ remains feasible.
The resulting aggregate flow may be decomposed into source--sink paths and
the paths ending at $n_i$ assigned to bidder $i$; source indifference makes
this assignment valid because no bidder-specific origin constraint remains.
Applying this cancellation pipe by pipe proves the claim.
\end{proof}

\begin{corollary}[Netting is inert in $\Rr$]\label{cor:netting-inert}
In $\Rr$ the netting and directed-arc models attain the same optimum.
\end{corollary}

\begin{proof}
Every directed-arc feasible solution is netting-feasible: if its flows on
$a$ are $x_a^+,x_a^-$, then
$z_a=x_a^+ -x_a^-\le x_a^+\le\bar q_a^+$ and
$-z_a\le x_a^-\le\bar q_a^-$. Netting can only weakly reduce its transport
cost, so $\wval_{\rm net}\ge\wval_{\rm dir}$. Conversely, Lemma~\ref{lem:noopposing}
converts \emph{any} netting-feasible solution, including a netting optimum,
into a directed-arc feasible solution of weakly greater welfare. Thus
$\wval_{\rm dir}\ge\wval_{\rm net}$, proving equality.
\end{proof}

\subsection{Restricted winner determination}\label{subsec:wdp-restricted}

We may therefore formulate the restricted WDP using directed arcs. Add a
super-source $Z^\ast$ and an uncapacitated source arc $(Z^\ast,s)$ for every
$s\in N_s$. Let
\[
E^\ast=\vec E\cup\{(Z^\ast,s):s\in N_s\},
\]
and define the per-unit arc cost by
\[
c_e=
\begin{cases}
c_e^t, & e\in\vec E,\\
c^s_s, & e=(Z^\ast,s),\ s\in N_s.
\end{cases}
\]
Write $f_e\ge0$ for the single-commodity flow on every arc $e\in E^\ast$,
and let $\delta^-(n)$ and $\delta^+(n)$ denote the incoming and outgoing arcs
of $n$ in the extended graph. The source arcs are uncapacitated, while the
transport arcs retain their capacities $\bar q_e$. For each coalition
$S\subseteq I$, the restricted WDP is the linear program

\begin{align}
\wval(S)=\max_{d,f}\quad
& \sum_{i\in S}\sum_{b\in\mathcal B_i}\unitval_b d_b
  -\sum_{e\in E^\ast} c_e f_e
\label{eq:wdp-r}\\
\text{s.t.}\quad
& \sum_{e\in\delta^+(n)}f_e-\sum_{e\in\delta^-(n)}f_e
  +\sum_{\substack{i\in S\\ n_i=n}}D_i
  -\mathbf 1[n=Z^\ast]\sum_{i\in S}D_i
  =0
&& \forall n\in N\cup\{Z^\ast\}
\quad(\pi_n),\notag\\
& f_e\le\bar q_e
&& \forall e\in\vec E
\quad(\lambda_e),\notag\\
& d_b\le q_b
&& \forall i\in S,\ \forall b\in\mathcal B_i
\quad(\alpha_b),\notag\\
& \sum_{b\in\mathcal B_i}d_b\le Q_i^T
&& \forall i\in S
\quad(\beta_i),\notag\\
& d,f\ge0. \notag
\end{align}

{Parentheses denote the corresponding dual variables. Thus there is a free conservation dual $\pi_n$ at \emph{every} node,
including $Z^\ast$, sources, and intermediate nodes. A licence names no inlet:
the single-commodity flow chooses the source. LP~\eqref{eq:wdp-r} defines
the restricted coalition value and, by Corollary~\ref{cor:netting-inert}, also
the value under netting.}

\subsection{Submodularity and core-stability}\label{subsec:submod-restricted}

\begin{lemma}[Submodularity of the coalition value]\label{lem:submod}
In $\Rr$, the coalition value $\wval$ of \eqref{eq:wdp-r} is monotone and
submodular over $I$.
\end{lemma}

\begin{proof}
By Corollary~\ref{cor:netting-inert} we use the directed-arc formulation; all
bidders share the single-commodity flow. For $r\in\R$ set
\[
   \phi_i(r)=\max_{d_i\in P_i}\{V_i(d_i)-rD_i\},\qquad D_i=\sum_{b\in\mathcal B_i}d_b,
\]
which is nonincreasing in $r$. The LP dual of \eqref{eq:wdp-r} is
\begin{equation}\label{eq:DS}
\begin{aligned}
\min_{\pi,\lambda,\alpha,\beta}\quad
 & \sum_{e\in\vec E}\bar q_e\lambda_e
   +\sum_{i\in S}\Bigl(\sum_{b\in\mathcal B_i}q_b\alpha_b+Q^T_i\beta_i\Bigr)\\
\text{s.t.}\quad
 & \lambda_e+\pi_u-\pi_v\ge -c^t_e, && e=(u,v)\in\vec E,\\
 & \pi_{Z^\ast}-\pi_s\ge -c^s_s, && s\in N_s,\\
 & \alpha_b+\beta_i+\pi_{n_i}-\pi_{Z^\ast}\ge \unitval_b, && i\in S,\ b\in\mathcal B_i,\\
 & \lambda_e,\alpha_b,\beta_i\ge 0.
\end{aligned}
\end{equation}
Normalise $\pi_{Z^\ast}=0$. Minimising over $\lambda_e\ge0$ contributes
$\bar q_e[\pi_v-\pi_u-c_e^t]_+$ to the reduced objective, and we may choose
$\lambda_e=[\pi_v-\pi_u-c_e^t]_+$ (also when $\bar q_e=0$). The source arcs give $\pi_s\le c^s_s$.

For fixed node potentials $\pi$, the dual variables $\alpha_b$ and $\beta_i$
associated with bidder $i$ solve the separate optimisation problem
\[
\begin{aligned}
\min_{\alpha,\beta}\quad
& \sum_{b\in \mathcal B_i} q_b\alpha_b + Q_i^T\beta_i \\
\text{s.t.}\quad
& \alpha_b+\beta_i \ge \unitval_b-\pi_{n_i},
\qquad b\in \mathcal B_i,\\
& \alpha_b\ge 0,\qquad \beta_i\ge 0.
\end{aligned}
\tag{$D_i(\pi_{n_i})$}
\]
But this is exactly the LP dual of
$\max
\left\{
V_i(d_i)-\pi_{n_i}D_i
\mid d_i\in P_i \right\}$.
Therefore, by strong duality,
\[
\min_{\alpha_b,\beta_i}
\left\{
\sum_{b\in \mathcal B_i} q_b\alpha_b + Q_i^T\beta_i\mid (\alpha,\beta) \in D_i(\pi_{n_i})
\right\}
=
\max_{d_i\in P_i}
\left\{
V_i(d_i)-\pi_{n_i}D_i
\right\}= \phi_i(\pi_{n_i}).
\]
Thus, after minimizing over the bidder-specific dual variables
$\alpha_b$ and $\beta_i$, bidder $i$ contributes simply
$\phi_i(\pi_{n_i})$ to the reduced dual objective. Hence

\begin{equation}
\label{eq:w-min}
w(S)
=
\min_{\pi\in\Pi}
\left\{
H(\pi)+\sum_{i\in S}\phi_i(\pi_{n_i})
\right\},
\end{equation}

where

$$
H(\pi)
=
\sum_{e=(u,v)}
\bar q_e[\pi_v-\pi_u-c_e^t]_+,
\qquad
\Pi=
\{\pi:\pi_{Z^\ast}=0,\ \pi_s\le c^s_s\ \forall s\in N_s\}.
$$

For two potential vectors $\pi,\rho$, let $\pi\wedge\rho$ and $\pi\vee\rho$ denote their componentwise minimum and maximum. We first note that
\begin{equation}\label{eq:H-submod}
   H(\pi)+H(\rho)\ge H(\pi\wedge\rho)+H(\pi\vee\rho).
\end{equation}
It suffices to check one arc $e=(u,v)$. If $\pi,\rho$ are ordered the same way at
$u$ and $v$ the two sides coincide; otherwise, say $\pi_u\le\rho_u$ and
$\pi_v\ge\rho_v$, the left arguments of $[\cdot]_+$ are
$\pi_v-\pi_u-c^t_e,\ \rho_v-\rho_u-c^t_e$ and the right ones
$\rho_v-\pi_u-c^t_e,\ \pi_v-\rho_u-c^t_e$; the two pairs have equal sum and the
left pair majorises the right, so convexity of $[\cdot]_+$ gives the arc
inequality, and summing yields \eqref{eq:H-submod}.

Now let $D_S(\pi)=H(\pi)+\sum_{i\in S}\phi_i(\pi_{n_i})$. Since each $\phi_i$ is
nonincreasing, for all $A,B\subseteq I$ and $\pi,\rho\in\Pi$,
\begin{equation}\label{eq:D-submod}
   D_A(\pi)+D_B(\rho)\ge D_{A\cap B}(\pi\wedge\rho)+D_{A\cup B}(\pi\vee\rho):
\end{equation}
bidders in $A\cap B$ contribute the same two terms on each side; for
$i\in A\setminus B$, $\phi_i(\pi_{n_i})\ge\phi_i(\max\{\pi_{n_i},\rho_{n_i}\})$,
symmetrically for $B\setminus A$; the $H$ terms obey \eqref{eq:H-submod}. Note
$\Pi$ is closed under $\wedge,\vee$.

The primal \eqref{eq:wdp-r} is feasible ($d=0$) and bounded ($P_i$ bounded,
capacities finite), so strong duality holds and the outer minimum in
\eqref{eq:w-min} is attained; let $\pi^A,\pi^B$ attain it for $A,B$. Then
\[
   \wval(A)+\wval(B)=D_A(\pi^A)+D_B(\pi^B)
   \ge D_{A\cap B}(\pi^A\wedge\pi^B)+D_{A\cup B}(\pi^A\vee\pi^B)
   \ge \wval(A\cap B)+\wval(A\cup B),
\]
using \eqref{eq:D-submod} and $\wval(S)\le D_S(\pi)$ for every $\pi\in\Pi$. Thus
$\wval$ is submodular; monotonicity holds since a bidder may deliver zero.
\end{proof}

\begin{proposition}[VCG is core-stable in $\Rr$]\label{prop:vcg-core}
In the restricted model the coalition value $\wval$ of \eqref{eq:wdp-r} is
submodular over bidders. Consequently the VCG payment vector lies in the core and
is its bidder-optimal point \citep{AusubelMilgrom2002,AusubelMilgrom2006}.
\end{proposition}

\section{Consequences beyond source indifference}
\label{sec:general-results}

We now use the general OD formulation of Section~\ref{sec:general} to
compare it with route enumeration, exhibit failures of VCG core stability, and
derive a dual pricing rule.

\subsection{Route enumeration is a source of inefficiency}\label{subsec:inefficiency}

\begin{example}[Diamond with an un-enumerated alternate route]\label{ex:diamond}
Consider Figure~\ref{fig:diamond}. Node $1$ hosts a source $s$ (price $0$).
Capacities are $1$ on every edge. Transfer costs: $c^t=0$ on
$a{:}1{\to}2$, $c{:}2{\to}4$, $e{:}2{\to}5$, and $c^t=1$ on
$b{:}1{\to}3$, $d{:}3{\to}4$. Consumer $C_1$ at node $4$ and consumer $C_2$ at
node $5$ each demand $1$ unit, with values $v_1=10$, $v_2=8$. Both cheapest
routes ($C_1$ via $\{a,c\}$, $C_2$ via $\{a,e\}$) traverse the shared edge $a$.

\begin{figure}[h]
\centering
\begin{tikzpicture}[>=Stealth,node distance=16mm,
   every node/.style={circle,draw,minimum size=7mm,inner sep=0pt}]
  \node (n1) {$1$};
  \node (n2) [right=of n1] {$2$};
  \node (n3) [below=of n1] {$3$};
  \node (n4) [below=of n2] {$4$};
  \node (n5) [right=of n2] {$5$};
  \draw[->] (n1) -- node[above,draw=none,fill=none,inner sep=1pt]{$a\,(0)$} (n2);
  \draw[->] (n1) -- node[left, draw=none,fill=none,inner sep=1pt]{$b\,(1)$} (n3);
  \draw[->] (n2) -- node[right,draw=none,fill=none,inner sep=1pt]{$c\,(0)$} (n4);
  \draw[->] (n3) -- node[below,draw=none,fill=none,inner sep=1pt]{$d\,(1)$} (n4);
  \draw[->] (n2) -- node[above,draw=none,fill=none,inner sep=1pt]{$e\,(0)$} (n5);
  \node[draw=none] [left=5mm of n1] {$s$};
  \node[draw=none] [below=3mm of n4] {$C_1$};
  \node[draw=none] [right=4mm of n5] {$C_2$};
\end{tikzpicture}
\caption{Example~\ref{ex:diamond}. Edge labels show name and transfer cost.
Both consumers' cheapest routes share edge $a$; $C_1$ has an alternate route
$\{b,d\}$ of cost $2$.}
\label{fig:diamond}
\end{figure}
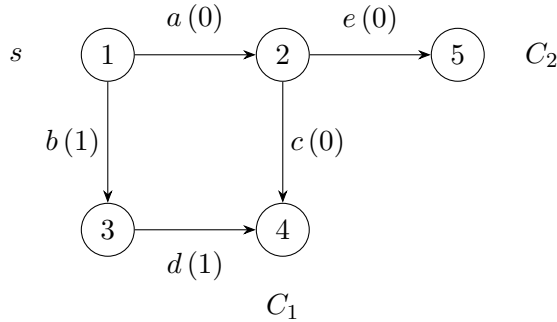

\paragraph{Route-CCA under cheapest-route enumeration} If each bidder
enumerates only its cheapest route, both name a bundle containing $a$. The
clearing solves $\max\{10x_1+8x_2:\ x_1+x_2\le 1\}$ and serves only $C_1$, resulting in a welfare of $10$ units, and $C_2$ is left unserved although the network could carry it.

\paragraph{OD mechanism} With OD bids $C_1{:}(1{\to}4,10)$ and
$C_2{:}(1{\to}5,8)$, the WDP routes $C_1$ over the costlier $\{b,d\}$ (net
$10-2=8$) and $C_2$ over $\{a,e\}$ (net $8$), serving both resulting in a  welfare of $16$ units.

The gap of $6$ is caused solely by bounded enumeration: had $C_1$ also listed $\{b,d\}$
the route-CCA would reach $16$ too. Since the number of routes is exponential,
no fixed enumeration is safe in general; the OD mechanism attains the
complete-enumeration optimum at polynomial size, because the WDP performs the
substitution across paths that the bidder would otherwise have to anticipate.
\end{example}

\begin{remark}
Example~\ref{ex:diamond} illustrates the efficiency deficit implied by the \emph{bidding language} and nothing else. In this particular case
its VCG payments are in fact core-stable (the contested resource is a single
edge, a substitution rather than a complementarity conflict). The
\emph{payment} defect requires a different structure, exhibited in
Section~\ref{subsec:dedicated}.
\end{remark}

\subsection{Dedicated sources: VCG may leave the core}\label{subsec:dedicated}

We now show that the source-indifference restriction in
Proposition~\ref{prop:vcg-core} cannot be dropped from its guarantee: in
the general model, VCG can be strictly outside the core. The mechanism is
complementarity. In $\Rr$ the no-opposing-flows lemma (Lemma~\ref{lem:noopposing})
rests on source indifference. Consumer-dependent access removes that
indifference and can therefore reintroduce the complementarity exhibited below;
inlet-specific caps, as in Example~\ref{ex:submodfail}, are another way to do so.

\begin{example}[Series network with dedicated entry]\label{ex:series}
Series network $1\to2\to3$ with edges $e_1{:}1{\to}2$, $e_2{:}2{\to}3$, each of
capacity $1$ and transfer cost $0$ (Figure~\ref{fig:series}). Node $1$ carries
gas usable by all; node $2$ carries a \emph{dedicated} entry point available
only to the bidder positioned to inject there (a distinct incoming pipeline or
liquefied natural gas (LNG) berth, not an open source). Three unit licences of value $2$:
\[
   X=(1\to3),\qquad Y=(1\to2),\qquad Z=(2\to3).
\]
$X$ needs both edges; $Y$ needs $e_1$ alone; $Z$ needs $e_2$ alone. Because the
node-$2$ entry is dedicated, $X$ cannot bypass $e_1$ by sourcing at node $2$ --
the maneuver that, in $\Rr$, would dissolve the complementarity.

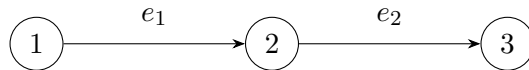
\begin{figure}[h]
\centering
\begin{tikzpicture}[>=Stealth,
   every node/.style={circle,draw,minimum size=7mm,inner sep=0pt}]
  \node (n1) {$1$};
  \node (n2) [right=24mm of n1] {$2$};
  \node (n3) [right=24mm of n2] {$3$};
  \draw[->] (n1) -- node[above,draw=none,fill=none,inner sep=1pt]{$e_1$} (n2);
  \draw[->] (n2) -- node[above,draw=none,fill=none,inner sep=1pt]{$e_2$} (n3);
\end{tikzpicture}
\caption{Example~\ref{ex:series}. $X$ bids $1\to3$ (both edges), $Y$ bids
$1\to2$, $Z$ bids $2\to3$.}
\label{fig:series}
\end{figure}

\paragraph{WDP} With $f_X+f_Y\le 1$ on $e_1$ and $f_X+f_Z\le1$ on $e_2$,
maximising $2(f_X+f_Y+f_Z)$ gives $f_Y=f_Z=1$, $f_X=0$: welfare $4$.

\paragraph{VCG} Removing $Y$, the others' best is $2$ (either $X$ takes the pair
or $Z$ takes $e_2$); with $Y$ present they already obtain $2$, so $Y$'s
externality is $0$. By symmetry $Z$ pays $0$. Total VCG revenue is $0$.

\paragraph{Core violation} The coalition $\{0,X\}$ can guarantee itself $2$
($X$ would pay up to $2$ for the two edges), so the core requires
$p_Y+p_Z\ge 2$. VCG collects $0<2$: the VCG outcome is \emph{strictly outside
the core}. Here $\wval$ is not submodular -- $X$'s two edges are complements --
which is possible precisely because the node-$2$ entry is dedicated, so
Lemma~\ref{lem:noopposing} and hence Proposition~\ref{prop:vcg-core} do not
apply.
\end{example}

\begin{proposition}[The source-indifference boundary]\label{prop:necessity}
A strict VCG core violation cannot occur in the restricted model $\Rr$.
Once source indifference is dropped, it may occur through consumer-dependent
source access or prices, or through inlet-specific quantity restrictions.
Consumer-dependent source access can therefore produce a failure, but is not
necessary for one.
\end{proposition}

\begin{proof}
Proposition~\ref{prop:vcg-core} rules out a violation in $\Rr$, whereas
Example~\ref{ex:series} supplies one with dedicated access. For inlet-specific
caps, use Example~\ref{ex:submodfail}, where additionally
$\wval(Y,Z)=2$. At the grand-coalition optimum bidders $Y,Z$ are served and
$X$ is not. Their VCG payments are

\[
 p_Y=\wval(X,Z)-(2-1)=0,
 \qquad
 p_Z=\wval(X,Y)-(2-1)=\tfrac12.
\]
The auctioneer therefore receives $1/2$, while the unserved bidder $X$
has utility zero. But coalition $\{0,X\}$ can attain
$\wval(X)=1>1/2$, so it blocks. This failure uses zero, hence
consumer-independent, source prices; it is caused solely by the
inlet-specific caps.
\end{proof}

\subsection{Dual (nodal) pricing}\label{subsec:dual}

Since VCG can leave the core in the general model, and since core stability is a
natural requirement for a regulated allocation (no coalition of the regulator
with a group of shippers should be able to improve on the outcome), we adopt a
payment rule that targets the core directly, following the core-selecting
philosophy of \citet{DayMilgrom2008}. The rule reads prices off the dual of the
WDP.

\subsubsection{The rule}

Let $(\pi^\star,\lambda^{+\star},\lambda^{-\star})$ be an optimal dual of
\eqref{eq:wdp}. The conservation duals induce \emph{nodal prices}: for a licence
$b\in\mathcal B_i$ with inlet $o_b$, delivered to bidder $i$'s sink $n_i$, let
\[
   \mu^\star_b=c^s_{o_b,i}+\pi^\star_{i,n_i}-\pi^\star_{i,o_b}
\]
be the delivered marginal price of a unit entering at $o_b$ and consumed at
$n_i$. Bidder $i$
pays
\begin{equation}\label{eq:lmp}
   p_i \;=\; \sum_{b\in\mathcal B_i} d^\star_b\,\mu^\star_b,
\end{equation}

the nodal spread on each licence times the quantity delivered on it. Of this
payment, the source and transfer components of $\mu^\star_b$ are forwarded to the
suppliers and the local TSOs (per the third-party-access split of
Section~\ref{sec:prelim}); the auctioneer retains only the congestion component,
so its total income (as shown later in Theorem~\ref{thm:dual-core}) is
\[
   \sum_i CC_i \;=\; \sum_{a\in\pipes} \bigl(\lambda^{+\star}_a\,\bar q^+_a
      +\lambda^{-\star}_a\,\bar q^-_a\bigr) \;\ge\; 0,
\]
the scarcity rent of the binding capacity constraints. In particular, when no
capacity binds ($\lambda^{+\star}=\lambda^{-\star}\equiv 0$) every payment passes through to
suppliers and TSOs and the auctioneer collects nothing -- a capacity auction
earns only the scarcity it actually rations, and is revenue-adequate on its
capacity account without ever being subsidised.

\begin{remark}[Why nodal, not per-edge]
Per-bidder flow on an edge is not determined by the WDP (equivalent optima
reshuffle routing among same-OD bidders), so a per-edge charge is ill-posed --
and under netting it can even change sign with counterflow. The nodal charge
\eqref{eq:lmp} depends only on the accepted OD quantities and the node prices;
for any selected optimal dual, the charge is decomposition-invariant. This is
the locational-marginal-pricing convention of electricity markets adapted to
capacity licences.
\end{remark}

\subsubsection{Core-stability}

\begin{remark}[Intuition]\label{rem:ce-intuition}
Core stability of the dual prices rests on a competitive-equilibrium argument.
By LP strong duality, an optimal allocation and an optimal dual satisfy
complementary slackness. At the nodal prices $\mu^\star$, each bidder's accepted
bundle maximises its surplus $V_i(d_i)-\mu^\star\!\cdot d_i$ over $P_i$, and the
optimal flow maximises the network's profit, i.e.\ its revenue at the nodal
prices net of source and transfer costs, which equals the congestion rent. That
the dual of a linear allocation program supplies equilibrium prices supporting
a core outcome is the package-assignment view of combinatorial markets
\citep{bikhchandani2002package}, rooted in the assignment game, where
competitive prices coincide with the core \citep{shapley1971assignment}. The
welfare of any coalition can be split into the surplus of its members at the
equilibrium prices and the network's profit at the same prices; neither part
can exceed its equilibrium value, so no coalition can block.
Theorem~\ref{thm:dual-core} makes this argument precise for both capacity
conventions. Since it relies only on dual feasibility and complementary
slackness, it holds for every optimal dual and requires no restriction on
source access. The same argument applies to the restricted problem
\eqref{eq:wdp-r}, where the nodal price paid by bidder $i$ is
$\pi^\star_{n_i}$.
\end{remark}

\noindent For bidder $i$ we write
$\mu^\star\!\cdot d_i=\sum_{b\in\mathcal B_i}\mu^\star_b d_b$, the
licence-indexed inner product, so that $p_i=\mu^\star\!\cdot d^\star_i$; for a
delivery profile $y=(y_i)_{i\in I}$, $\mu^\star\!\cdot y=\sum_i\mu^\star\!\cdot y_i$.

\begin{theorem}[Dual pricing is core-stable]\label{thm:dual-core}
In the general model, under either the netting or the directed-arc capacity
convention, let $(d^\star,f^\star)$, together with $g^\star$ under netting, be
an optimal WDP allocation, and choose any optimal dual of the corresponding
WDP. The payments \eqref{eq:lmp}
obtained from its conservation potentials support a core outcome.
\end{theorem}

\begin{proof}
We give the duals and the complementary-slackness argument explicitly.
Throughout the proof, $i\in I$, $b\in\mathcal B_i$, and
$a\in\pipes$ has reference orientation $a^+=(u_a,v_a)$ and reverse
orientation $a^-=(v_a,u_a)$. The dual variables of the atomic and total bid
bounds are $\alpha_b\ge0$ and $\beta_i\ge0$. The conservation potentials
$\pi_{i,n}$ are free. Define

\begin{equation}\label{eq:mu-dual-proof}
 \mu_b(\pi)=c^s_{o_b,i}+\pi_{i,n_i}-\pi_{i,o_b}.
\end{equation}

Under netting, the exact dual of the grand-coalition instance of
\eqref{eq:wdp} is

\begin{equation}\label{eq:dual-netting}
\begin{aligned}
(D_{\rm net})\qquad
\min_{\pi,\tau,\lambda,\alpha,\beta}\quad
 &\sum_{a\in\pipes}
     (\bar q_a^+\lambda_a^++\bar q_a^-\lambda_a^-)
  +\sum_{i\in I}\left(
       \sum_{b\in\mathcal B_i}q_b\alpha_b+Q_i^T\beta_i\right)\\
\mathrm{s.t.}\quad
 &\alpha_b+\beta_i+\mu_b(\pi)\ge \unitval_b
       &&\forall i,\ b\in\mathcal B_i,\\
 &\tau_a+\pi_{i,u_a}-\pi_{i,v_a}\ge0
       &&\forall i,a,\\
 &-\tau_a+\pi_{i,v_a}-\pi_{i,u_a}\ge0
       &&\forall i,a,\\
 &\lambda_a^+-\tau_a\ge-c_a^{t,+}
       &&\forall a,\\
 &\lambda_a^-+\tau_a\ge-c_a^{t,-}
       &&\forall a,\\
 &\lambda^+,\lambda^-,\alpha,\beta\ge0,
       \qquad \pi,\tau\ \text{free}.
\end{aligned}
\end{equation}
Indeed, the five displayed constraint groups are respectively the dual constraints associated with the primal variables $d_b,f_{i,a}^+,f_{i,a}^-,g_a^+$, and $g_a^-$. Notice
that the second and third inequalities are opposites. Hence every feasible
netting dual satisfies

\begin{equation}\label{eq:tau-potential}
 \tau_a=\pi_{i,v_a}-\pi_{i,u_a}\qquad\forall i,a.
\end{equation}

Under separate directed-arc capacities, the net-flow variables $g^\pm$
and $\tau$ are absent, transport cost is paid on the bidder flows, and the
exact dual is

\begin{equation}\label{eq:dual-directed}
\begin{aligned}
(D_{\rm dir})\qquad
\min_{\pi,\lambda,\alpha,\beta}\quad
 &\sum_{a\in\pipes}
     (\bar q_a^+\lambda_a^++\bar q_a^-\lambda_a^-)
  +\sum_{i\in I}\left(
       \sum_{b\in\mathcal B_i}q_b\alpha_b+Q_i^T\beta_i\right)\\
\mathrm{s.t.}\quad
 &\alpha_b+\beta_i+\mu_b(\pi)\ge \unitval_b
       &&\forall i,\ b\in\mathcal B_i,\\
 &\lambda_a^++\pi_{i,u_a}-\pi_{i,v_a}\ge-c_a^{t,+}
       &&\forall i,a,\\
 &\lambda_a^-+\pi_{i,v_a}-\pi_{i,u_a}\ge-c_a^{t,-}
       &&\forall i,a,\\
 &\lambda^+,\lambda^-,\alpha,\beta\ge0,
       \qquad \pi\ \text{free}.
\end{aligned}
\end{equation}

The primal is feasible (take every variable equal to zero), and its
objective is bounded above because every $P_i$ is bounded and all costs are
nonnegative. Hence both formulations have optimal primal and dual solutions,
strong duality holds, and any optimal primal--dual pair satisfies
complementary slackness. Fix such a pair and henceforth add $\star$ to all of
its variables.

\paragraph{Bidder optimality}
At prices $\mu^\star_b=\mu_b(\pi^\star)$, bidder $i$'s maximum attainable
surplus is

\begin{equation}\label{eq:bidder-surplus}
 U_i(\mu^\star)=
 \max_{d_i\in P_i}\sum_{b\in\mathcal B_i}(\unitval_b-\mu_b^\star)d_b.
\end{equation}
The exact dual of this bidder problem is

\begin{equation}\label{eq:bidder-surplus-dual}
 \min_{\alpha,\beta}
 \left\{\sum_{b\in\mathcal B_i}q_b\alpha_b+Q_i^T\beta_i:
 \alpha_b+\beta_i\ge \unitval_b-\mu_b^\star,\quad
 \alpha_b,\beta_i\ge0\right\}.
\end{equation}
This is exactly bidder $i$'s block of either global dual. The relevant
complementary-slackness equations are

\begin{align}
 \alpha_b^\star(q_b-d_b^\star)&=0,
 &\beta_i^\star\left(Q_i^T-\sum_b d_b^\star\right)&=0,
 \label{eq:cs-bid-caps}\\
 d_b^\star
 \bigl(\alpha_b^\star+\beta_i^\star+\mu_b^\star-\unitval_b\bigr)&=0.
 \label{eq:cs-bid-demand}
\end{align}
They show that $d_i^\star$ and
$((\alpha_b^\star)_{b\in\mathcal B_i},\beta_i^\star)$ are
feasible for \eqref{eq:bidder-surplus} and \eqref{eq:bidder-surplus-dual} and
have equal objectives. Consequently $d_i^\star$ solves the bidder problem, and

\begin{equation}\label{eq:bidder-utility}
 u_i=V_i(d_i^\star)-\mu^\star\!\cdot d_i^\star
     =U_i(\mu^\star)\ge0,
\end{equation}
where nonnegativity follows because $d_i=0$ is feasible.

\paragraph{Network optimality under netting}
Call a nonnegative profile $y=(y_i)_{i\in I}$ \emph{deliverable} if the
network constraints admit flows that inject and withdraw its quantities; the
bid-region constraints $y_i\in P_i$ are not imposed on this network production
set. Let $y$ be any such profile, and
let $\widehat f,\widehat g$ be any netting-feasible flows delivering it.
Multiplying bidder $i$'s conservation equations by $\pi_{i,n}^\star$ and
summing over nodes gives the identity

\begin{equation}\label{eq:potential-identity}
 \sum_{b\in\mathcal B_i}
   (\pi_{i,n_i}^\star-\pi_{i,o_b}^\star)y_b
 =\sum_{a\in\pipes}\left[
   (\pi_{i,v_a}^\star-\pi_{i,u_a}^\star)\widehat f_{i,a}^+
  +(\pi_{i,u_a}^\star-\pi_{i,v_a}^\star)\widehat f_{i,a}^-\right].
\end{equation}
Using \eqref{eq:mu-dual-proof}, \eqref{eq:tau-potential}, and the net-flow
identity, the network's revenue net of source and transfer cost is therefore

\begin{align}
 &\mu^\star\!\cdot y
 -\sum_{i,b}c^s_{o_b,i}y_b
 -\sum_a(c_a^{t,+}\widehat g_a^+
             +c_a^{t,-}\widehat g_a^-)\notag\\
 &=\sum_a\left[(\tau_a^\star-c_a^{t,+})\widehat g_a^+
              +(-\tau_a^\star-c_a^{t,-})\widehat g_a^-\right]\notag\\
 &\le\sum_a(\lambda_a^{+\star}\widehat g_a^+
             +\lambda_a^{-\star}\widehat g_a^-)
 \le\sum_a(\lambda_a^{+\star}\bar q_a^+
             +\lambda_a^{-\star}\bar q_a^-).
 \label{eq:network-bound-netting}
\end{align}
The first inequality is precisely the last two constraint groups of
$D_{\rm net}$, and the second uses primal capacity feasibility. At the optimal
flow both inequalities are equalities. Indeed, complementary slackness gives

\begin{align*}
 g_a^{+\star}(\lambda_a^{+\star}-\tau_a^\star+c_a^{t,+})&=0,
 &\lambda_a^{+\star}(\bar q_a^+-g_a^{+\star})&=0,\\
 g_a^{-\star}(\lambda_a^{-\star}+\tau_a^\star+c_a^{t,-})&=0,
 &\lambda_a^{-\star}(\bar q_a^--g_a^{-\star})&=0.
\end{align*}

\paragraph{Network optimality under directed capacities}
For any directed-feasible $(y,\widehat f)$, the same potential identity and
the two arc constraints of $D_{\rm dir}$ yield

\begin{align}
 &\mu^\star\!\cdot y
 -\sum_{i,b}c^s_{o_b,i}y_b
 -\sum_{i,a}(c_a^{t,+}\widehat f_{i,a}^+
             +c_a^{t,-}\widehat f_{i,a}^-)\notag\\
 &\le\sum_a\left(
      \lambda_a^{+\star}\sum_i\widehat f_{i,a}^+
     +\lambda_a^{-\star}\sum_i\widehat f_{i,a}^-\right)
 \le\sum_a(\lambda_a^{+\star}\bar q_a^+
             +\lambda_a^{-\star}\bar q_a^-).
 \label{eq:network-bound-directed}
\end{align}
For example, the coefficient of $\widehat f_{i,a}^+$ on the first line is
$\pi_{i,v_a}^\star-\pi_{i,u_a}^\star-c_a^{t,+}$, which is at most
$\lambda_a^{+\star}$ by $D_{\rm dir}$; the reverse direction is identical.
Complementary slackness,

\begin{align*}
 f_{i,a}^{+\star}
 (\lambda_a^{+\star}+\pi_{i,u_a}^\star-\pi_{i,v_a}^\star+c_a^{t,+})&=0,
 &\lambda_a^{+\star}
 (\bar q_a^+-\sum_i f_{i,a}^{+\star})&=0,\\
 f_{i,a}^{-\star}
 (\lambda_a^{-\star}+\pi_{i,v_a}^\star-\pi_{i,u_a}^\star+c_a^{t,-})&=0,
 &\lambda_a^{-\star}
 (\bar q_a^--\sum_i f_{i,a}^{-\star})&=0,
\end{align*}
makes both inequalities in \eqref{eq:network-bound-directed} equalities at
the optimal flow.

Thus, under either convention, if $C(y)$ denotes the minimum source and
transport cost of delivering $y$, then

\begin{equation}\label{eq:network-profit}
 R(\mu^\star):=\max_{y\ \mathrm{deliverable}}
       \{\mu^\star\!\cdot y-C(y)\}
 =\sum_{a\in\pipes}
   (\lambda_a^{+\star}\bar q_a^+
    +\lambda_a^{-\star}\bar q_a^-).
\end{equation}
The maximum is attained by $d^\star$ and its optimal flow. Since payments
cover source and transfer costs as pass-throughs, the auctioneer retains exactly
this network profit:

\begin{equation}\label{eq:auctioneer-profit}
 u_0=\sum_i p_i-C(d^\star)=R(\mu^\star)\ge0.
\end{equation}

\paragraph{The core inequalities}
Equations \eqref{eq:bidder-utility} and \eqref{eq:auctioneer-profit}, along
with strong duality, give payoff feasibility:

\[
 u_0+\sum_{i\in I}u_i
 =\sum_{i\in I}V_i(d_i^\star)-C(d^\star)
 =\wval(I).
\]
Now fix a coalition $S\subseteq I$ and any profile
$d=(d_i)_{i\in S}$ that it can deliver, setting $d_i=0$ outside $S$. Then

\begin{align*}
 \sum_{i\in S}V_i(d_i)-C(d)
 &=\sum_{i\in S}
     \bigl(V_i(d_i)-\mu^\star\!\cdot d_i\bigr)
   +\bigl(\mu^\star\!\cdot d-C(d)\bigr)\\
 &\le\sum_{i\in S}U_i(\mu^\star)+R(\mu^\star)
 =\sum_{i\in S}u_i+u_0.
\end{align*}
The first bound uses bidder optimality separately for every $i\in S$ and
network optimality for the coalition's deliverable profile. Maximising the left
side over all coalition-feasible profiles gives

\[
 \wval(S)\le u_0+\sum_{i\in S}u_i,
\]
which is the core constraint. Together with payoff feasibility and
individual rationality, this proves the result in the general model.

\end{proof}

\begin{remark}[Degeneracy]\label{rem:degen}
Theorem~\ref{thm:dual-core} shows every optimal dual yields a \emph{core} point;
under degeneracy the dual is non-unique, so the \emph{payment vector} is not
unique -- a choice \emph{among} core points, not a threat to core membership. A
canonical selection (lexicographic, or the bidder-optimal core point
of \citet{DayMilgrom2008}, computed for instance by the quadratic
core-selecting rule of \citet{day2012quadratic}) can be layered on top for
determinacy and fairness.
\end{remark}

\subsubsection{Illustrative examples}

We reuse the series network of Example~\ref{ex:series} and then exhibit the
counter-directed (credit) case.
In the following examples, we work with reduced WDP formulations obtained by eliminating flow variables through conservation equations and omitting redundant constraints. These formulations are exact projections of the corresponding full WDPs. Moreover, the optimal reduced dual solutions considered below admit objective-preserving extensions to the full WDP duals, as verified by the explicit construction of conservation potentials. Consequently, the resulting dual prices are valid for the original market.

\begin{example}[Dual pricing repairs Example~\ref{ex:series}]
\label{ex:series-fixed}

Let us first recall the primal problem  written in full in the case of example \label{ex:series-fixed}.
Using $x,y,z\in[0,1]$ for the accepted quantities (each also the flow of
that bidder), the WDP (after dropping redundant constraints) becomes
\begin{equation}\label{eq:wdp-series}
\begin{aligned}
   \wval(I)=\max\ & 2x+2y+2z\\
   \text{s.t.}\ & x+y\le 1\ \ (\lambda_1),\qquad x+z\le 1\ \ (\lambda_2),\qquad
   0\le x,y,z\le 1\ \ (\alpha_X,\alpha_Y,\alpha_Z).
\end{aligned}
\end{equation}
As discussed earlier, the unique optimum is
\[
   x^\star=0,\qquad y^\star=z^\star=1,\qquad \wval(I)=4 .
\]

The LP dual of \eqref{eq:wdp-series} is
\begin{equation}\label{eq:dual-series}
\begin{aligned}
\min\ & \lambda_1+\lambda_2+\alpha_X+\alpha_Y+\alpha_Z\\
\text{s.t.}\ & \alpha_X+\lambda_1+\lambda_2\ge 2 \qquad \\
             & \alpha_Y+\lambda_1\ge 2, \qquad\qquad\ \
             & \alpha_Z+\lambda_2\ge 2 \qquad\qquad\ \\
             & \lambda,\alpha\ge 0 .
\end{aligned}
\end{equation}
Using the optimal primal, complementary slackness determines the
optimal duals:
\begin{enumerate}[label=(\roman*),leftmargin=2.2em]
\item $y^\star=1>0$: the $y$-constraint is tight, $\alpha_Y=2-\lambda_1$; since
$\alpha_Y\ge 0$, $\lambda_1\le 2$. Likewise $z^\star>0$ gives
$\alpha_Z=2-\lambda_2$ and $\lambda_2\le 2$.
\item $x^\star=0<1$: the cap of $X$ is slack, so $\alpha_X=0$, and dual
feasibility for $x$ becomes $\lambda_1+\lambda_2\ge 2$.
\item Both capacity constraints are tight at the primal optimum, so
$\lambda_1,\lambda_2$ may be positive.
\end{enumerate}
For any such $(\lambda_1,\lambda_2)$ the dual objective equals
$\lambda_1+\lambda_2+0+(2-\lambda_1)+(2-\lambda_2)=4=\wval(I)$, so all of them are
optimal. Hence the set of optimal duals is
\begin{equation}\label{eq:dualset-series}
   \Lambda^\star=\bigl\{(\lambda_1,\lambda_2):\ 0\le\lambda_1\le 2,\
   0\le\lambda_2\le 2,\ \lambda_1+\lambda_2\ge 2\bigr\},
\end{equation}
with $\alpha_X=0$, $\alpha_Y=2-\lambda_1$, $\alpha_Z=2-\lambda_2$. In words: each
edge price is at most $2$ because the accepted bidder who uses it alone values
it at $2$, and the two prices together are at least $2$
because the rejected bidder $X$, who needs both, values the pair at $2$. The dual is degenerate: the primal optimum is a vertex at
which more constraints are tight than there are variables.

Let us use the following two observations

\begin{itemize}
    \item O1: Complementary
slackness in the general directed WDP for a flow variable says: if bidder $i$ ships a positive amount on
arc $(u,v)$, then
\[
   \pi_{i,v}-\pi_{i,u}=c^t_{(u,v)}+\lambda_{(u,v)}.
\]
Summing along the path a
unit actually uses, the potential difference between inlet and
sink equals \emph{total transfer cost plus total congestion price along the
path}. Hence
\[
   \mu_b=\underbrace{c^s_{o_b,i}}_{\text{supplier}}
        +\underbrace{\textstyle\sum_{e\in\text{path}}c^t_e}_{\text{TSOs}}
        +\underbrace{\textstyle\sum_{e\in\text{path}}\lambda_e}_{\text{capacity charge}} .
\]
Under netting, the dual forces one common price gap $\tau_a$ per physical pipe
$a$ for all bidders, and the same decomposition holds with $\tau_a$ in place of
$c^t+\lambda$ (a unit travelling \emph{against} the reference orientation picks
up $-\tau_a$).
    \item O2: For each licence, the dual has
$\alpha_b\ge 0$ (shadow value of its cap $q_b$) with
$\alpha_b+\beta_i+\mu_b\ge \unitval_b$, and complementary slackness gives:
\begin{itemize}[leftmargin=1.4em]
\item a licence accepted at a positive level satisfies
$\unitval_b=\mu_b+\alpha_b+\beta_i$, hence $\mu_b\le\unitval_b$ (no total caps
bind in the examples, so $\beta_i=0$);
\item a licence accepted strictly below its cap has $\alpha_b=0$, hence
$\mu_b\ge\unitval_b$ if rejected and $\mu_b=\unitval_b$ if partially accepted.
\end{itemize}
\end{itemize}

\paragraph*{Potentials and nodal prices}

All costs are zero, so by O1 a common potential can be chosen for all
bidders:
\[
   \pi_1=0,\qquad \pi_2=\lambda_1,\qquad \pi_3=\lambda_1+\lambda_2 .
\]
(Bidder-specific potentials could differ only for $X$, who ships nothing and
pays nothing, so this choice is without loss for the payments.) The nodal prices
are
\[
   \mu_X=\pi_3-\pi_1=\lambda_1+\lambda_2,\qquad
   \mu_Y=\pi_2-\pi_1=\lambda_1,\qquad
   \mu_Z=\pi_3-\pi_2=\lambda_2 ,
\]
with all source prices zero. O2 is visible directly: the rejected $X$ faces
$\mu_X\ge 2=\unitval_X$; the accepted $Y,Z$ face $\mu_Y,\mu_Z\le 2$.

\paragraph*{Payments, utilities and revenue}

Only $Y$ and $Z$ receive gas, so
\[
   p_Y=\lambda_1,\qquad p_Z=\lambda_2,\qquad p_X=0,
\]
and since there are no source or transfer costs, the payments are pure
capacity charges. Utilities and auctioneer revenue:
\[
   u_X=0,\qquad u_Y=2-\lambda_1,\qquad u_Z=2-\lambda_2,\qquad
   u_0=p_Y+p_Z=\lambda_1+\lambda_2 ,
\]
and $u_0=\lambda_1\cdot 1+\lambda_2\cdot 1$ is exactly the congestion rent
$\sum_e\lambda_e\bar q_e$. Payoffs add up to the welfare:
$u_0+u_X+u_Y+u_Z=(\lambda_1+\lambda_2)+0+(2-\lambda_1)+(2-\lambda_2)=4$.

\paragraph*{The core}

Write $p_Y,p_Z$ for the payments. Since the payoffs sum to $4$, the core
constraints $u_0+\sum_{i\in S}u_i\ge\wval(S)$ read:
\begin{center}
\renewcommand{\arraystretch}{1.2}
\begin{tabular}{l|c|c|l}
\toprule
coalition with $0$ & $u_0+\sum_{i\in S}u_i$ & $\wval(S)$ & condition\\
\midrule
$\{0\}$ & $p_Y+p_Z$ & 0 & $p_Y+p_Z\ge 0$\\
$\{0,X\}$ & $p_Y+p_Z$ & 2 & $\boldsymbol{p_Y+p_Z\ge 2}$\\
$\{0,Y\}$ & $2+p_Z$ & 2 & $p_Z\ge 0$\\
$\{0,Z\}$ & $2+p_Y$ & 2 & $p_Y\ge 0$\\
$\{0,X,Y\}$ & $2+p_Z$ & 2 & $p_Z\ge 0$\\
$\{0,X,Z\}$ & $2+p_Y$ & 2 & $p_Y\ge 0$\\
$\{0,Y,Z\}$ & 4 & 4 & ---\\
$\{0,X,Y,Z\}$ & 4 & 4 & --- (feasibility)\\
\bottomrule
\end{tabular}
\end{center}
Together with individual rationality ($u_Y,u_Z\ge0$, i.e.\ $p_Y,p_Z\le 2$), the
core in payment space is
\[
   \mathcal C=\bigl\{(p_Y,p_Z):\ 0\le p_Y\le 2,\ 0\le p_Z\le 2,\ p_Y+p_Z\ge 2\bigr\}.
\]
Comparing with \eqref{eq:dualset-series}: under $p_Y=\lambda_1$, $p_Z=\lambda_2$
\emph{the set of dual-price outcomes coincides with the core} in this example
(Figure~\ref{fig:core}). The only binding economic condition is the one for
$\{0,X\}$: the auctioneer together with the rejected bidder could sell the two
edges to $X$ for $2$, so $Y$ and $Z$ together must pay at least $2$.

\paragraph{Why VCG fails here} VCG charges $Y$ its externality. Without $Y$, the others obtain
$\wval(\{X,Z\})=2$; with $Y$ present they obtain
$\wval(I)-2=2$; hence $p_Y^{\rm VCG}=0$, and symmetrically
$p_Z^{\rm VCG}=0$. The point $(0,0)$ violates $p_Y+p_Z\ge 2$: the coalition
$\{0,X\}$ blocks. Each of $Y$ and $Z$ is individually not pivotal (removing
one of them, $X$ still cannot be served because the other occupies an edge), but
jointly they displace $X$. Dual prices cannot make this mistake, because the
price of the path $1\to3$ is $\mu_X=\lambda_1+\lambda_2$, which must cover
$X$'s value whenever $X$ is rejected.

\begin{figure}[h]
\centering
\begin{tikzpicture}[scale=1.4]
  \fill[gray!25] (0,2) -- (2,0) -- (2,2) -- cycle;
  \draw[->] (0,0) -- (2.6,0) node[right] {$p_Y=\lambda_1$};
  \draw[->] (0,0) -- (0,2.6) node[above] {$p_Z=\lambda_2$};
  \draw (0,2) -- (2,0);
  \draw (2,0) -- (2,2) -- (0,2);
  \foreach \t in {1,2} {\draw (\t,0.05) -- (\t,-0.05) node[below] {\small $\t$};
                        \draw (0.05,\t) -- (-0.05,\t) node[left] {\small $\t$};}
  \fill (0,0) circle (1.5pt) node[below left] {\small VCG};
  \fill (1,1) circle (1.5pt) node[below left,xshift=2pt] {\small symmetric};
  \fill (2,2) circle (1.5pt) node[above right] {\small seller-optimal};
  \node at (1.45,1.45) {$\mathcal C=\Lambda^\star$};
\end{tikzpicture}
\caption{Series example: the core $\mathcal C$ (shaded) and the set of
dual-price outcomes $\Lambda^\star$ coincide. The segment $p_Y+p_Z=2$ is the
set of bidder-optimal core points; VCG lies outside.}
\label{fig:core}
\end{figure}
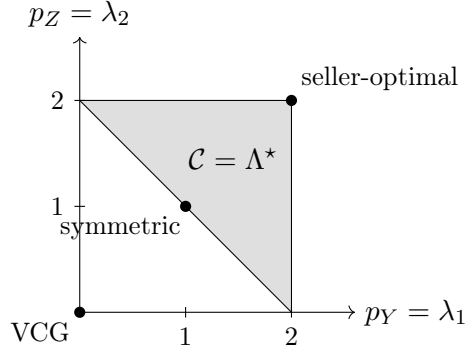



\end{example}

\begin{example}[A consumer earns a credit under netting]
\label{ex:credit}

Let us assume a line $1-2-3$, pipes $e_1=\{1,2\}$ and $e_2=\{2,3\}$ with reference orientations
$1\to2$ and $2\to3$ respectively, zero transfer costs, and netting convention.
$q^+_{e_1}=\bar q^-_{e_1}=10$ units (it will be uncongested), while  $\bar q^+_{e_2}=\bar q^-_{e_2}=1$. An open source at node $1$ sells gas at price $2$.
\begin{center}
\begin{tabular}{lllccc}
\toprule
bidder & inlet & sink & valuation $\unitval$ & source price & net value\\
\midrule
$A$ & $1$ & $3$ & $10$ & $2$ & $8$\\
$B$ & $1$ & $3$ & $8$ & $2$ & $6$\\
$C$ & $1$ & $3$ & $6$ & $2$ & $4$\\
$D$ & $3$ (own gas) & $1$ & $3$ & $0$  & $3$\\
\bottomrule
\end{tabular}
\end{center}
Each licence has quantity $1$. $D$ is a positioned shipper: it holds its own gas
at node $3$, whose acquisition cost is sunk, so its source price in the model is
$c^s_{3,D}=0$; the node-$3$ gas is \emph{not} an open source for $A,B,C$.
\end{example}

\begin{figure}[h]
\centering
\begin{tikzpicture}[>=Stealth,
   every node/.style={circle,draw,minimum size=7mm,inner sep=0pt}]
  \node (n1) {$1$};
  \node (n2) [right=28mm of n1] {$2$};
  \node (n3) [right=28mm of n2] {$3$};
  \draw[->] ([yshift=2pt]n1.east) -- node[above,draw=none,inner sep=1pt]{\small $A,B$: $2$ units} ([yshift=2pt]n2.west);
  \draw[<-] ([yshift=-2pt]n1.east) -- node[below,draw=none,inner sep=1pt]{\small $D$: $1$ unit} ([yshift=-2pt]n2.west);
  \draw[->] ([yshift=2pt]n2.east) -- node[above,draw=none,inner sep=1pt]{\small $A,B$: $2$ units} ([yshift=2pt]n3.west);
  \draw[<-] ([yshift=-2pt]n2.east) -- node[below,draw=none,inner sep=1pt]{\small $D$: $1$ unit} ([yshift=-2pt]n3.west);
  \node[draw=none,circle=false,rectangle,minimum size=0pt,anchor=south,align=center] at ([yshift=11mm]n1.north) {\small source, price $2$\\\small $D$'s sink};
  \node[draw=none,rectangle,minimum size=0pt,anchor=south,align=center] at ([yshift=11mm]n3.north) {\small sinks of $A,B,C$\\\small $D$'s gas};
  \node[draw=none,rectangle,minimum size=0pt,anchor=north] at ([yshift=-11mm]n2.south) {\small net flow on $e_2$: $2-1=1=\bar q^+_{e_2}$};
\end{tikzpicture}
\caption{Optimal gross flows. Under netting only the net forward flow
$2-1=1$ counts against the capacity of $e_2$.}
\label{fig:credit}
\end{figure}
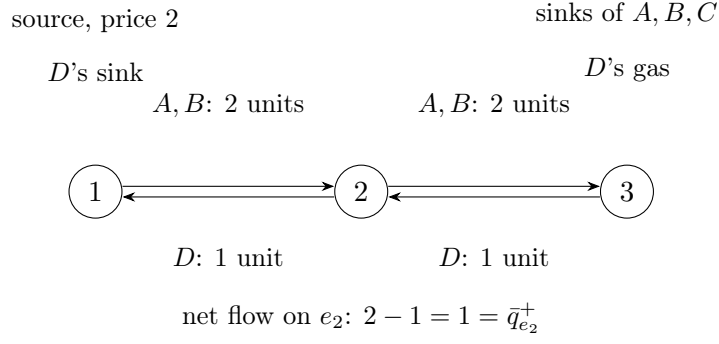

\paragraph*{The winner-determination LP and its solution}

Write $a,b,c,d\in[0,1]$ for the accepted quantities. Source payments are
pass-throughs, so the welfare objective uses net values. Under netting, the
capacity constraint of $e_2$ applies to the \emph{net} forward flow, i.e.\ the
forward shipments of $A,B,C$ minus the backward shipment of $D$:
\begin{equation}\label{eq:wdp-credit}
\begin{aligned}
   \wval(I)=\max\ & 8a+6b+4c+3d\\
   \text{s.t.}\ & a+b+c-d\le 1\ \ (\lambda),\qquad 0\le a,b,c,d\le 1\ \ (\alpha_A,\dots,\alpha_D)
\end{aligned}
\end{equation}
(the backward-capacity constraint $d-(a+b+c)\le\bar q^-_{e_2}$ and the
constraints of $e_1$ never bind and are omitted).

Accepting $D$ is worth $3$ directly \emph{and} raises the room for forward
shipments from $1$ to $2$. The best use of two forward units is $A$ and $B$
($8+6$), so
\[
   a^\star=b^\star=d^\star=1,\qquad c^\star=0,\qquad
   \wval(I)=8+6+3=17 .
\]
(Without $D$, only one forward unit fits and it goes to $A$: welfare $8$.) The relevant
coalition values are as
\begin{center}
\begin{tabular}{l|ccccccccc}
\toprule
$S$ & $\{A\}$ & $\{B\}$ & $\{D\}$ & $\{A,B\}$ & $\{A,D\}$ & $\{B,D\}$ &
$\{A,B,C\}$ & $\{A,B,D\}$ & $I$\\
\midrule
$\wval(S)$ & 8 & 6 & 3 & 8 & 11 & 9 & 8 & 17 & 17\\
\bottomrule
\end{tabular}
\end{center}
Note that $\wval$ is not submodular: the marginal value of $B$ is $0$ when added
to $\{A\}$ but $6$ when added to $\{A,D\}$. $B$ and $D$ are complements, because
$D$'s counterflow creates the room that $B$ uses.

\paragraph*{The dual and all its optimal solutions}

The dual of \eqref{eq:wdp-credit} is
\begin{equation}\label{eq:dual-credit}
\begin{aligned}
\min\ & \lambda+\alpha_A+\alpha_B+\alpha_C+\alpha_D\\
\text{s.t.}\ & \alpha_A+\lambda\ge 8,\quad \alpha_B+\lambda\ge 6,\quad
              \alpha_C+\lambda\ge 4,\quad \alpha_D-\lambda\ge 3,\quad
              \lambda,\alpha\ge0 .
\end{aligned}
\end{equation}
Note the sign in $D$'s constraint: $d$ enters the capacity constraint with
coefficient $-1$. Complementary slackness gives:
\begin{enumerate}[label=(\roman*),leftmargin=2.2em]
\item $a^\star,b^\star,d^\star>0$: their constraints are tight, so
$\alpha_A=8-\lambda$, $\alpha_B=6-\lambda$, $\alpha_D=3+\lambda$. Nonnegativity of
$\alpha_B$ gives $\lambda\le 6$.
\item $c^\star=0<1$: $\alpha_C=0$, and feasibility for $c$ gives $\lambda\ge 4$.
\end{enumerate}
The dual objective is then
$\lambda+(8-\lambda)+(6-\lambda)+0+(3+\lambda)=17=\wval(I)$ for every
$\lambda\in[4,6]$. Hence
\begin{equation}\label{eq:lamset-credit}
   \lambda^\star\in[4,6]:
\end{equation}
the congestion price lies between the net value of the best \emph{rejected}
forward bidder ($C$, $4$) and that of the worst \emph{accepted} one ($B$, $6$),
the usual uniform-price interval. Without $D$ the optimum serves only
$A$ and the optimal duals are $\lambda\in[6,8]$ (between $B$ and $A$).

Without $D$ the marginal
forward unit is $B$'s, which is rejected, so $\lambda\ge 6$. With $D$, $B$ is
accepted and $C$ becomes the marginal rejected bidder, so the price range moves
down to $[4,6]$. More capacity (here created by counterflow) means a lower
scarcity price.

\paragraph*{Potentials and nodal prices}

In the netting dual of the general WDP, each physical pipe $a$ has one free
price gap $\tau_a=\pi_{i,v_a}-\pi_{i,u_a}$, common to all bidders, with
$\lambda^+_a-\tau_a\ge -c^{t,+}_a$ and $\lambda^-_a+\tau_a\ge -c^{t,-}_a$, and
complementary slackness on the net-flow variables $g^\pm_a$:
\begin{itemize}[leftmargin=1.4em]
\item Pipe $e_1$: net forward flow $g^+_{e_1}=2-1=1>0$, below capacity, so
$\lambda^+_{e_1}=0$ and $\tau_{e_1}=c^{t,+}+\lambda^+_{e_1}=0$.
\item Pipe $e_2$: net forward flow $g^+_{e_2}=1=\bar q^+_{e_2}$, so
$\tau_{e_2}=c^{t,+}+\lambda^+_{e_2}=\lambda$ with $\lambda\in[4,6]$; the backward
direction is slack, $\lambda^-_{e_2}=0$.
\end{itemize}
Normalising $\pi_{i,1}=0$ for each bidder gives $\pi_{i,2}=0$ and
$\pi_{i,3}=\lambda$ for all $i$. The nodal prices are
\[
\begin{aligned}
   \mu_A=\mu_B=\mu_C &= c^s_{1,\cdot}+\pi_{\cdot,3}-\pi_{\cdot,1}=2+\lambda,\\
   \mu_D &= c^s_{3,D}+\pi_{D,1}-\pi_{D,3}=0+0-\lambda=-\lambda .
\end{aligned}
\]
$D$'s unit travels against the reference orientation of $e_2$, so it picks up
$-\tau_{e_2}=-\lambda$: the reversed spread.

\paragraph*{Payments, decomposition and revenue}

\begin{center}
\renewcommand{\arraystretch}{1.2}
\begin{tabular}{l|c|ccc|c}
\toprule
bidder & payment $p_i=\mu_i$ & to supplier & to TSOs & capacity charge
 & utility $u_i$\\
\midrule
$A$ & $2+\lambda$ & 2 & 0 & $\lambda$ & $10-(2+\lambda)=8-\lambda$\\
$B$ & $2+\lambda$ & 2 & 0 & $\lambda$ & $8-(2+\lambda)=6-\lambda$\\
$C$ & 0 (rejected) & 0 & 0 & 0 & 0\\
$D$ & $-\lambda$ & 0 & 0 & $-\lambda$ & $3-(-\lambda)=3+\lambda$\\
\midrule
total & & 4 & 0 & $\lambda$ & \\
\bottomrule
\end{tabular}
\end{center}
With $\lambda=4$: $A$ and $B$ each pay $6$ (of which $4$
is the capacity charge), and $D$ is paid $4$, so its utility is $3+4=7$.

\paragraph*{Revenue adequacy} The auctioneer keeps the capacity charges:
$u_0=\lambda+\lambda-\lambda=\lambda$. This equals the congestion rent
$\lambda^+_{e_2}\bar q^+_{e_2}=\lambda\cdot 1$, as the dual-pricing rule
guarantees in general. For $\lambda=4$: $4+4-4=4$. The credit to $D$ is financed
out of the forward shippers' capacity charges; supplier and TSO payments are
untouched.

\paragraph*{Feasibility} $u_0+\sum_i u_i=\lambda+(8-\lambda)+(6-\lambda)+0+
(3+\lambda)=17=\wval(I)$ for every $\lambda\in[4,6]$.

\paragraph*{The core}

For each $\lambda\in[4,6]$ the dual-price payoffs are
$u=(8-\lambda,\,6-\lambda,\,0,\,3+\lambda)$, $u_0=\lambda$.
All $16$ core inequalities hold for every $\lambda\in[4,6]$,
as theorem \ref{thm:dual-core} for dual pricing predicts. At $\lambda=4$ the coalitions $\{0,A\}$ and
$\{0,B\}$ are tight.

\paragraph*{Comparison with VCG} Here VCG also pays $D$, and more. The VCG
capacity charges are
\[
\begin{aligned}
CC_A&=\wval(\{B,C,D\})-(17-8)=13-9=4,\\
CC_B&=\wval(\{A,C,D\})-(17-6)=15-11=4,\\
CC_C&=\wval(\{A,B,D\})-17=0,\\
CC_D&=\wval(\{A,B,C\})-(17-3)=8-14=-6 .
\end{aligned}
\]
So $D$ receives $6$ under VCG, but the auctioneer collects only $4+4-6=2$, and
the coalition $\{0,A\}$ blocks: $u_0+u_A=2+(8-4)=6<8=\wval(\{A\})$ (so do
$\{0,B\}$ and $\{0,C\}$). This is the complementarity between $B$ and $D$: VCG rewards $D$ with its full marginal contribution, which
includes $B$'s value, and the total is too generous. Under the dual prices $D$
receives $\lambda\in[4,6]$, and the auctioneer keeps the congestion rent.

\section{Discussion}\label{sec:discussion}

The results presented above leave several design choices open, and we discuss
the most important of them in turn.

A first set of questions concerns incentives. Core stability and
dominant-strategy incentive compatibility cannot, in general, be obtained
simultaneously: VCG secures the latter but, as Example~\ref{ex:series} shows,
not necessarily the former, whereas dual pricing secures the core at the price
of dominant-strategy incentive compatibility. Under the payment
rule~\eqref{eq:lmp}, truthful bidding is optimal for a price-taking bidder,
that is, for a bidder who treats the nodal prices $\mu^\star$ as given. This is
the natural benchmark in markets with many participants, but a large bidder
that is pivotal on a congested edge can, in principle, influence the dual
prices through its bids. When the dual is degenerate and the payment vector is
therefore not unique (Remark~\ref{rem:degen}), selecting the bidder-optimal core
point \citep{DayMilgrom2008,day2012quadratic} minimises the residual incentive
to deviate and is the natural refinement of the rule. Whether such a canonical
core selection can be obtained in closed form from the duals of the
winner-determination problem, rather than by solving a secondary optimisation
problem, remains an open question.

A second design choice concerns counter-directed flows. Throughout the paper we
have carried two capacity conventions, netting and separate directed arcs. In
the restricted model the two conventions coincide
(Corollary~\ref{cor:netting-inert}), so the choice is immaterial there. In the
general model, netting yields weakly higher welfare but admits negative
payments, that is, credits for counterflow (Example~\ref{ex:credit}), whereas
the directed-arc convention excludes such credits at some welfare cost. The
latter is arguably closer to regulatory practice, since EU entry--exit capacity
is booked per directed point. As the nodal pricing rule is well defined under
both conventions, we make the choice explicit rather than resolve it, and leave
it to the regulator.

The scope of the mechanism is determined by the nature of the transported good.
Replacing route enumeration by routing within the winner-determination problem
is sound because gas is fungible: any unit delivered to a bidder's sink is as
good as any other, so the mechanism is free to choose and split paths. Where
the good is an intrinsically rigid bundle -- a train path occupying a specific
sequence of track segments and time slots \citep{Stojadinovic2019}, or an
all-or-nothing wavelength -- the relaxation underlying the OD language is no
longer exact: the route bundle constitutes a genuine complementarity, the
winner-determination problem acquires integrality constraints, and both the
exactness of the dual prices and the core-stability guarantee may be lost. Our
results are therefore sharpest for divisible and substitutable transport
services such as gas or bandwidth.

Finally, the two shortcomings of the convex combinatorial auction of
\citet{Csercsik2022} addressed in this paper share a common origin. Route-based
bidding hard-codes bundles that the mechanism should be free to form or
dissolve: a route is an indivisible combination of edges, and an
inlet-specific bid fixes the origin of the gas. The OD language dissolves the
route bundle and thereby closes the efficiency gap illustrated in
Example~\ref{ex:diamond}. Departures from source indifference, whether through
dedicated sources or through inlet-specific caps, may expose a residual
complementarity between bidders; these are precisely the situations in which
VCG may leave the core, and in which the flow-conservation duals restore core
stability.

\section{Conclusion and future work}\label{sec:conclusion}

This paper analyses the problem of transport capacity allocation in
capacity-constrained networks through auctions in which participants bid for
delivery rather than for individual capacity products. Addressing the flaws of the convex
combinatorial auction framework presented in \citet{Csercsik2022}, we showed that requiring
bidders to enumerate their routes in advance may lead to strictly inefficient
allocations, and we proposed an origin--destination bidding language in which
bidders specify where gas may enter and where it is to be delivered, while the
routing is determined by a network-flow linear program. Because this
winner-determination problem remains a linear program of polynomial size, the
language attains the welfare of complete route enumeration without its
combinatorial burden.

On the payment side, we identified source indifference as the boundary of a
positive result. When sources are openly accessible at uniform prices and
bidders value only the quantity delivered to their own sink, some optimal
aggregate flows never use a pipe in both directions, the winner-determination
problem reduces to a single-commodity flow, and an elementary argument based on
node potentials shows that the coalition value is submodular in the set of
bidders; consequently, the VCG outcome lies in the core. Once bidders face
consumer-specific access to sources or inlet-specific quantity caps, this
property may fail and VCG payments may lie strictly outside the core. For this
general setting we proposed a dual (nodal) pricing rule derived from the shadow
prices of the flow-conservation constraints, and proved that every optimal dual
yields a core outcome under both capacity conventions considered, with the
auctioneer's revenue equal to the congestion rent of the binding capacity
constraints.

\subsection{Future work}

Several directions remain open. First, the analysis has been illustrated on
stylised examples; applying the mechanism to a realistic network, such as the
European gas transmission system, would require the estimation of inverse
demand functions and would allow a quantitative comparison with current
allocation practice. Second, the incentive properties of dual pricing deserve
further study, including the extent to which large bidders can influence nodal
prices and the computation of a canonical core point. In addition, the present model
treats consumers as the only strategic participants; including producers on the
supply side leads to a combinatorial double auction over origin--destination
licences, which may be considered as a natural extension of the current work.
Furthermore, the model used in the current work assumed a steerable medium, an assumption that does not hold for electricity-type infrastructure networks, where the physical properties of the network elements determine the resulting flows for a given input-output pattern. Studying the problem of auction-based capacity allocation in Kirchoff-type networks can also mean a plausible continuation of the current wok.

\section*{Acknowledgements}
D\'{a}vid Csercsik thanks Robert Day for discussions about the flaws and limitations of the
original convex combinatorial auction \citep{Csercsik2022} during the Dagstuhl 2025 Seminar series.

\section*{Funding}
Gergely Csáji was supported by the Ministry of Innovation and Technology of Hungary from the National
Research, Development and Innovation Fund, by grants STARTING 153093 and FK 145838.

\section*{CRediT authorship contribution statement}
\textbf{Gergely Cs\'aji:} Methodology, Formal analysis, Validation, Writing - Original Draft, Writing - Review \& Editing, Visualization. \textbf{D\'avid Csercsik:} Conceptualization, Methodology, Validation, Writing - Original Draft, Writing - Review \& Editing.

\section*{Declaration of competing interest}
The authors declare that they have no known competing financial interests or
personal relationships that could have appeared to influence the work reported
in this paper.

\section*{Declaration of generative AI and AI-assisted technologies in the manuscript preparation process}
During the preparation of this work the authors used Claude (Anthropic) to support literature research, language improvement and formatting of the manuscript. After using this tool, the authors reviewed and edited the
content as needed and take full responsibility for the content of the
published article.

\bibliographystyle{apacite}
\bibliography{references}

\end{document}